\documentclass[9pt]{article}
\usepackage{spconf,amsmath,graphicx}
\usepackage{amssymb,amsfonts,eurosym,geometry,ulem,graphicx,caption,color,setspace,sectsty,comment,footmisc,caption,pdflscape,subfigure,array,hyperref}
\usepackage[utf8]{inputenc}
\usepackage{url}
\usepackage{multirow}
\usepackage{mathtools}
\newcommand{\mR}{\mathbb{R}}

\newcommand{\mc}{\mathcal}

\newcommand{\mE}{\mathbb{E}}

\newcommand{\om}{\omega}
\newcommand{\Om}{\Omega}

\renewcommand{\leq}{\leqslant}
\renewcommand{\geq}{\geqslant}
\mathtoolsset{showonlyrefs}

\usepackage{algpseudocode}
\usepackage{algorithmicx}
\usepackage{algorithm}
\usepackage{caption}
\usepackage{eufrak}

\newtheorem{prop}{Proposition}

\newtheorem{assume}{Assumption}

\DeclareMathOperator{\eps}{\epsilon}

\newcommand{\Ket}{|\Phi\rangle}
\newcommand{\Bra}{\langle \Phi |}

\newcommand{\ketkk}{|\varphi_k\rangle}

\newcommand{\ketetaa}{|\eta_j\rangle}

\newcommand{\braetaa}{\langle\eta_j|}

\newcommand{\tr}{\text{Tr}}

\newenvironment{proof}[1][Proof]{\noindent\textbf{#1.} }{\ \rule{0.5em}{0.5em}}

\newcolumntype{L}[1]{>{\raggedright\let\newline\\arraybackslash\hspace{0pt}}m{#1}}
\newcolumntype{C}[1]{>{\centering\let\newline\\arraybackslash\hspace{0pt}}m{#1}}
\newcolumntype{R}[1]{>{\raggedleft\let\newline\\arraybackslash\hspace{0pt}}m{#1}}
\begin{document}
\title{Detection in human-sensor systems under quantum prospect theory using Bayesian persuasion frameworks}
\name{Yinan Hu and Quanyan Zhu\address{Tandon School of Engineering, New York University, Brooklyn, NY, USA}}

\maketitle
\begin{abstract}
    Human-sensor systems have a wide range of applications in fields such as robotics, healthcare, and finance. These systems utilize sensors to observe the true state of nature and generate strategically designed signals, aiding humans in making more accurate decisions regarding the state of nature. We adopt a Bayesian persuasion framework that is integrated with quantum prospect theories. In this framework, we develop a detection scheme where humans aim to determine the true state by observing the realization of quantum states from the sensor. We derive the optimal signaling rule for the sensor and the optimal decision rule for humans. We discover that this scenario violates the total law of probability. Furthermore, we examine how the concepts of rationality can influence the human detection performance and the signaling rules employed by the sensor.
\end{abstract}

\section{Introduction}

Detection methods play a vital role in statistical signal processing, encompassing a wide range of applications, such as studying sensor attacks \cite{meira2020synthesis_sensor_attack}, analyzing internet traffic, and more. Within the realm of human-sensor systems, detection frameworks have emerged as essential components in domains like robotics \cite{amini2013quantum_robotics}, healthcare systems, and recommendation systems \cite{lu2012recommender}. An essential aspect of research within this field centers on the influence of sensors in guiding human decision-making processes through the meticulous design of signals intended for individuals.

Recent studies have embraced quantum decision theories \cite{busemeyer2012quantum_cognition} to interpret various phenomena in human cognitive science such as order effect \cite{trueblood2017quantum_inference} and violation of sure-thing-principle \cite{busemeyer2009quantum_violation_sure_thing_principple} that cannot be adequately explained using classical theories.  which refers to the violation of the total probability law when probabilities represent human perception in the decision-making process. In \cite{snow2022quickest_detection_QDT}, researchers have developed quickest detection frameworks by integrating quantum decision models, aiming to capture bounded rationalities observed in human decision-making. However, it is essential to consider risk preference \cite{tversky1992advances_prospect_theory} as a crucial factor in human decision-making. Humans may not assign equal weight to gains and losses, and risk-averse individuals may be unwilling to trade the possibility of a loss for the chance of a gain when selecting lotteries. Theories on risk measures \cite{artzner1999coherent_risk} have been developed to provide a more sophisticated characterization of human risk preferences.



To this end, we formulate the detection of sensor-human systems to using quantum decision theory \cite{sornette2020quantum_propsect_theory}. This decision model integrates classical outcomes and the psychological state to capture human's bounded rationality, including risk-preference, in the decision-making process. 
Our contributions can be summarized as follows.
First, we develop a comprehensive detection framework for human-sensor systems that takes into account risk-preference and incorporates interference effects,  thereby capturing the inherent bounded rationality of humans. 
Second, we establish the existence of an optimal policy resembling a likelihood-ratio test for the human receiver within the detection framework. This finding sheds light on the optimal decision-making strategy for the human component of the system, enhancing our understanding of their behavior. 

The rest of the paper is organized as follows: Section \ref{sec:formulation} presents the formulation of the relationship between the sensor and the human receiver, along with the protocol for communication between them regarding the true state of nature. In Section \ref{sec:sol_concept}, we discuss the optimal decision policies for the human receivers and the optimal signaling rules for the sensors. Section \ref{sec:numerical_results} is dedicated to the numerical simulation of the proposed solution concepts. Specifically, we verify the violation of the sure-thing principle and illustrate the optimal thresholding based on different prior beliefs about the true hypothesis. Finally, Section \ref{sec:conclusion} concludes the paper. 

\textbf{Related work:} This work builds upon previous research such as \cite{gezici2018_likelihood_ratio_test_prospect_hypo} and \cite{nadendla2016human_agent_hypo_testing}, but introduces a novel perspective by incorporating quantum decision theory into the human decision-making process. In addition, we employ the Bayesian persuasion model \cite{kamenica2019bayesian_persuasion} to formulate the behaviors of the sensor. By integrating these frameworks, we aim to capture the complex dynamics within the sensor-human system.


\textbf{Notation:}
$\mc{H}$: the Hilbert space (over the set of real numbers $\mR$); $\mc{H}^*$: the dual space of $\mc{H}$; $\langle \Phi|\in\mc{H}^*$: the left state vector; $|\Phi\rangle\in\mc{H}$: the right state vector;
$B(\mc{H})$: the space of all positive, Hermitian and bounded operators from $\mc{H}$ to itself; 
$\mc{S}$: the subset of $B(\mc{H})$ such as trace of its operators is $1$; $\mc{V}$: the space of projection-valued measurements \cite{von2018mathematical_QM};
$S$: the space of signals; $\Delta(\cdot)$: the set of probability measures over the given space;
$\mathbf{1}\in B(\mc{H})$: the identity operator; $\Om=\{H_0,H_1\}$: the space of states.

\section{Problem Formulation} 
\label{sec:formulation}
In this section, we develop the framework of detection in human-sensor systems where human adopts quantum decision theory \cite{sornette2020quantum_propsect_theory}. We assume that there are two underlying states of the system:  normal state $\om = H_0$  and abnormal state $\om = H_1$. Under each hypothesis, the observations generated $s'\in S$ obey different distributions:
\begin{equation}
    H_0: s'\sim  f_0(s),\;\;H_1: s'\sim f_1(s),\;
    \label{hypo_rho1_rho0} 
\end{equation}
where $f_0,f_1$ are probability density functions.  We associate a common prior $p(H_1),p(H_0)$, with $p(H_1)+p(H_0)=1$ with the true hypothesis. 
\begin{figure}
    \centering
    \includegraphics[scale=0.24]{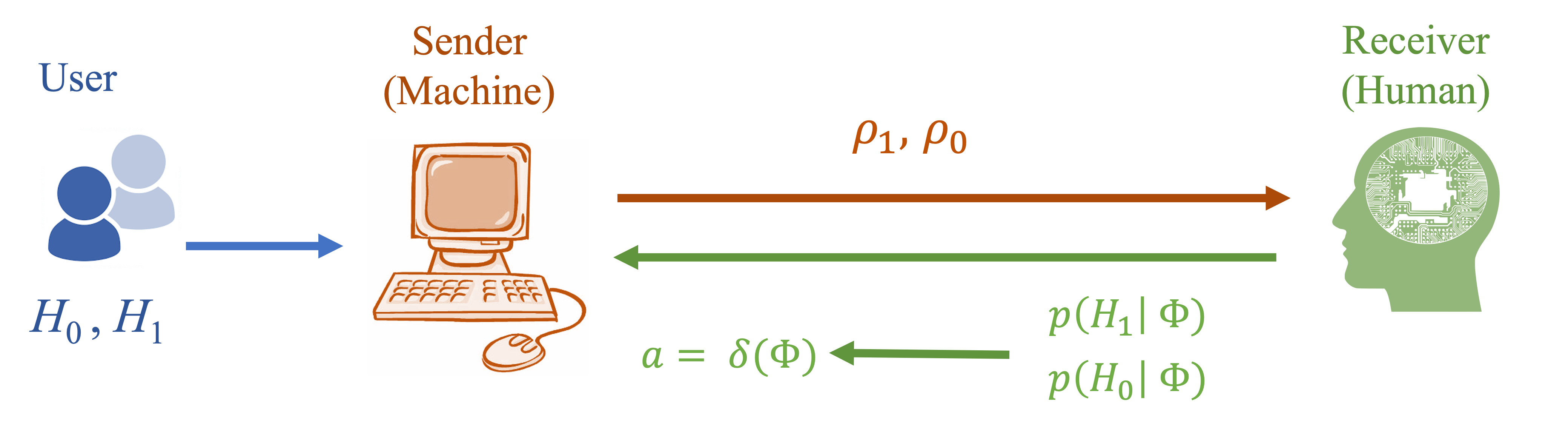}
    \caption{The human-sensor interaction scheme. Before the game starts the sender (sensor) commits to a type-dependent signaling devices $\rho_0,\rho_1\in \mc{S}$. The sender is a sensor/machine that obtains the true state. }
    \label{fig:quantum_prospect_theory}
\end{figure}

The key feature of quantum decision theory lies in the entanglement of quantum states, which connects classical, exogenous outcomes with the quantum states associated with the endogenous psychological mind state of humans. This connection is achieved by mathematically associating a composite state of mind, represented as $|\Phi\rangle \in \mathcal{H} = \mathcal{H}_C \otimes \mathcal{H}_I$, with each classical observation of the signal, denoted as $s' \in S$. The composite state consists of a quantum state $|\psi^{s'}\rangle \in \mathcal{H}_C$ representing the cognitive aspect of the human's mind and a quantum state $|\chi\rangle \in \mathcal{H}_I$ representing the subjective interpretation or perception associated with the signal. The human receiver produces a `prospect state'
$ |\Phi\rangle = |\psi^s\rangle \otimes \sum_{k}{a_{sk}\ketkk}= \sum_{k}{a_{sk}|\psi^s\varphi_k\rangle}$ (with  $\sum_{k=1}{|a_{sk}|^2} = 1$), where $\{\ketkk\}$  is  a set of orthonormal basis spanning $\mc{H}_I$ as the space of perception states. 
When the realization of signals $s$ is stochastic, obeying distributions $f_1,f_0$ as in \eqref{hypo_rho1_rho0} and depending on the true state of nature, the sensor (sender) generates two possible ``mixed prospect states" $\tilde{\rho}_1$ and $\tilde{\rho}_0$ from the set $\mathcal{S}$ as follows:  
 \begin{align}
        \tilde{\rho}_1 & = \sum_{s}{f_1(s)\Big(\sum_{k}{a_{sk}|\psi^s\varphi_k\rangle\Big)\Big(\sum_{k'}{a_{sk'}\langle \psi^s\varphi_{k'}|}}\Big)},  
        \label{eq:rho1_prospect_state}
        \\
        \tilde{\rho}_0 & =  \sum_{s}{f_0(s)\Big(\sum_{k}{a_{sk}|\psi^s\varphi_k\rangle\Big)\Big(\sum_{k'}{a_{sk'}\langle \psi^s\varphi_{k'}|}\Big)}}, 
       \label{eq:rho0_prospect_state}
    \end{align}

\textbf{The protocol:} We adopt a sender-receiver framework to model the relationship between a sensor and a human being, as illustrated in Figure \ref{fig:quantum_prospect_theory}. The system consists of an underlying state that satisfies the requirements described in \eqref{hypo_rho1_rho0}. We apply the following decision-making protocol for human-sensor system.  The sender(sensor) first commits to a type-dependent signal devices: $\rho_1,\rho_0\in \mc{S}$ based on the `vanilla prospect states' $\tilde{\rho}_1,\tilde{\rho}_0$. The sender observes the realization of the true state $\om\in \Om$. The sender delivers randomized signal to the receiver(human). The receiver  observes the realization of the signal $\Phi\in \mc{H}$. The receiver takes an action $a = \delta(\Phi)\in\{0,1\}$ suggesting that the human thinks the hypothesis $H_a$ holds true. The sender and the receiver both obtain reward/cost based on the true state and their actions.  

\textbf{Sender's utility function:} 
We denote $v:\Om\times \mc{S}\times \mc{S}\times B(\mc{H}) \rightarrow \mR$ be the sender's utility function.  The sender aims at seeking optimal signaling devices $\rho^*_1,\rho^*_0$ by solving the following optimization problem:
\begin{equation}
\begin{aligned}
    &\underset{\rho_1,\rho_0\in \mc{S}}{\max}\Big\{\mE_{\substack{\Phi_0\sim\rho_0 \\ \Phi_1\sim\rho_1}}[v(\om,\Phi_1,\Phi_0, P^*)]\Big\},\\
    &\text{s.t.}\;P^*\in \arg\underset{P\in B(\mc{H})}{\min}\;u(\rho^*_1,\rho_0^*,P),
\end{aligned}
\label{eq:Bayesian_Persusasion_sender}
\end{equation}
where $P^*\in B(\mc{H})$ and $u: \mc{S}\times \mc{S} \times B(\mc{H})\rightarrow \mR$ are the optimal decision rule and the utility function of the human receiver, respectively,  which we will specify later. 

\textbf{Human's decision model:}
Upon receiving the signal $s\in S$ from the sensor, the human receiver first construct a prospect state $\Phi$ based on $s$ as mentioned before. Then the human receiver updates the common prior belief $p(H_0),p(H_1)$ on the true hypothesis into posterior belief based on Bayes' rule:
\begin{equation}
    p(H_j|\Phi) = \frac{p(H_j)\Bra \rho_j \Ket}{p(H_0)\Bra \rho_0 \Ket + p(H_1)\Bra \rho_1 \Ket},\;\;j=0,1.
\end{equation}

The human arrives at a decision rule $\delta^*$ through projective positive-valued measurements (POVM) \cite{von2018mathematical_QM}: $P(\cdot) = \sum_{j}{\ketetaa\braetaa}$,
 where $\{\ketetaa\}_j\subset \mc{H}$ forms a set of orthonormal base vectors for the decision maker to find out.  Given any realization of the prospect state $\Ket \in \mc{H}$, we can form that the human receiver makes a probabilistic decision $a_1=\delta(\Phi) = 1$ (i.e., considering that $H_1$ holds true) with probability
$\mathbb{P}(a_1=1|\Phi) = \Bra P \Ket \equiv g + q$, where $g,q$ represent the \textit{utility factor} and the \textit{attraction factor} respectively \cite{sornette2020quantum_propsect_theory}.

Motivated by the frameworks in \cite{gezici2018_likelihood_ratio_test_prospect_hypo}, we let the probability $P_F = \tr(P\rho_0)$ denote the false alarm rate, and the probability $P_D = \tr(P\rho_1)$  the detection rate. They are used to characterize receiver's risk function due to errors. We now formulate human's problem as an optimization problem where we construct the human receiver's weighted risk function $u$ that takes into account the probabilities $P_F,P_D$ in a way similar to \cite{gezici2018_likelihood_ratio_test_prospect_hypo} based on sensor's equivalent signaling rules $\rho_1,\rho_0$ as follows:
\begin{equation}
\begin{aligned}
  &\underset{P\in B(\mc{H})}{\min }u(\rho_1,\rho_0,P,\Phi)  \\
  &={ w(p(H_0|\Phi)\tr(P \rho_0))}u_{01}+ w(p(H_1|\Phi)\tr(P\rho_1))u_{11} \\
  &+ w(p(H_0|\Phi)(1-\tr(P \rho_0)))u_{00}\\
     &+w(p(H_1|\Phi)(1-\tr (P\rho_1)))u_{10},
\end{aligned} 
    \label{eq:util_detector_proactive}
\end{equation}
where for convenience we assume $u_{11},u_{00}<0,u_{01},u_{10}>0$. The weight function $w:[0,1]\rightarrow [0,1]$ in \eqref{eq:util_detector_proactive} is selected the same as in \cite{nadendla2016human_agent_hypo_testing}:
\begin{equation}
    w(z;\eps) = {z^{\eps}},\;\;z\in[0,1],\;\eps>0,
    \label{eq:weight_function}
\end{equation}
 where $0<\eps<1$ corresponds to a pessimistic agent, while $\eps>1$ an optimistic agent \cite{nadendla2016human_agent_hypo_testing}. 

\section{Theoretical Results}
\label{sec:sol_concept}
In this section, we solve the optimization problems \eqref{eq:util_detector_proactive} and \eqref{eq:Bayesian_Persusasion_sender}. Notice the sender's utility function $v$  in \eqref{eq:Bayesian_Persusasion_sender} characterizes sender's type-dependent strategies of changing the original prospect states. We assume that the sensor construct mixed prospect states with the same form but with different `perception coefficients' $a^1_{sk}\in\mR,a^0_{sk}, s\in S,k=1,2,\dots, d$. Below, we demonstrate that the human agent's optimal decision rule corresponds to the Quantum Likelihood Ratio Test (QLRT), which bears resemblance to the approach outlined in \cite{helstrom1976quantum_estimation_detection}. 
\begin{prop}[QLRT as human's optimal strategy]
\label{prop:human_decision_rule}
Let $\rho_1,\rho_0\in B(\mc{H})$ be the sender's signaling devices and let $\Phi\in\mc{H}$ be the prospect state. Let the problem \eqref{eq:util_detector_proactive} be the human's receiver's optimization problem, where the receiver aims at developing optimal measurements $P^*\in B(\mc{H})$. Suppose that the weight function $w$ defined in \eqref{eq:weight_function} is monotonically increasing. Then, we arrive at the following conclusion: 
\begin{equation}
    P^* =  \sum_{\eta_j>0}{\ketetaa\braetaa},
    \label{eq:optimal_decision_rule_prospect}
\end{equation}
where $\ketetaa$ are the eigenvectors of $\rho_1-\tau \rho_0$ with eigenvalues $\eta_j,\;j=1,2,\dots$ i.e.,$(\rho_1-\tau \rho_0) \ketetaa = \eta_j \ketetaa$ for some $\tau\geq 0$.  
\end{prop}
\begin{proof}
We adopt the proof similar to the one for proposition 1 in \cite{gezici2018_likelihood_ratio_test_prospect_hypo}. We know that the human receiver aims to distinguish between two states $H_1,H_0$ corresponding to two mixed states $\rho_1,\rho_0$, respectively as in \eqref{eq:rho0_prospect_state} and \eqref{eq:rho1_prospect_state}.  Denote $y^*= \tr(P^*\rho_1)$ as the detection rate and $x^*=\tr(P^*\rho_0)$ as the false alarm rate. 
We construct the projective measurement $P^*$ for binary hypothesis testing as \eqref{eq:optimal_decision_rule_prospect}. 
Now we claim the decision rule $\delta^*:\mc{H}\rightarrow [0,1]$ is constructed as $\delta^*(\Phi) = \Bra P^*\Ket$ is optimal. To see this,
first notice via \cite{helstrom1976quantum_estimation_detection} that $P^*$ minimizes the Bayes risk for quantum detection:
\begin{equation}
    P^*\in\arg\underset{P\in B(\mc{H})}{\min}\; \tau\tr(P\rho_0) + \tr((\mathbf{1}-P)\rho_1).
\end{equation} 
Thus if we pick another arbitrary projective operator-valued measurement $P'\in B(\mc{H})$ leading to another detection rate $x'=\tr(P'\rho_1)$ and false alarm rate $y' = \tr(P'\rho_0)$,
similar to the proof of Neyman-Pearson lemma \cite{neyman_pearson1933}, we can derive
for any $P'\in B(\mc{H})$,$\tau(x^*-x')\leq y^*-y'$, where $x^*,y^*$ are the false alarm rate and detection rate of $P^*$, respectively. Since we set $y^*=y'$, we conclude $x^*\leq x'$. 
Since $w$ is monotone increasing and that the left hand side is $0$, we have 
\begin{equation}
    w(\tr((1-P^*)\rho_1))\geq w(\tr ((1- P')\rho_1)).
\end{equation}
As a result, for any measurement $P'$ leading to a certain detection rate $y'$, we can always construct a corresponding $P^*$ of the form \eqref{prop:human_decision_rule} achieving a lower false positive rate than the one of $P'$ under the same detection rate. Thus using $P^*$ of the form \eqref{eq:optimal_decision_rule_prospect} we can lower the second term without changing the first term in \eqref{eq:util_detector_proactive}. 
Thus a generic human's optimal decision rule $P^*$ minimizing the utility function $u$ must be of the form given in \eqref{eq:optimal_decision_rule_prospect}. 
\end{proof}


\section{Numerical results} 
\label{sec:numerical_results}

In this section, we present the numerical illustration of the optimal detection policy for the human agent and the optimal signaling rule for the sensor, as discussed in Section \ref{sec:sol_concept}. To demonstrate these concepts, we utilize a cognitive case study known as the Prisoner's Dilemma, as described in \cite{busemeyer2006quantum_disjunct}. In this scenario, there are two parties involved: a human agent and her opponent.
signaling rules according to $\rho_1,\rho_0$.
The human decision maker faces a binary choice: defection $(a=1)$ or cooperation $(a=0)$. Simultaneously, the opponent's choices, represented by the true state of nature, consist of defection $H_1$ or cooperation $H_0$. Notably, the human agent is unaware of the opponent's action until after she has made her own decision.
The human agent's objective is to maximize her reward, which is higher when her action aligns with that of her opponent. Human makes a decision based on \eqref{eq:util_detector_proactive} while the sensor (interpreted as a message passer) producing signaling rules according to $\rho_1,\rho_0$. 




\textbf{Violation of sure-thing-principle:}
The sure-thing-principle (STP), or total probability law,  can be interpreted in decision theory as a phenomenon that if under two states $H_1,H_0$, an action $a$ is preferred to $a'$, then such preference is carried over to the scenario where the state is unknown. 
Authors in \cite{busemeyer2006quantum_disjunct} have used quantum probabilistic models to justify such violations under investigations of the case of Prisoner's dilemma. 
We assume $d=2$ and $K=5$ and fix the realization of the state $\Ket$. We apply the same payoff matrix as in \cite{busemeyer2006quantum_disjunct} and set the reward values as $u_{00} = 20, u_{01} = 5, u_{10} = 10, u_{11} = 25$. Using different values of the attraction factor from $0$ to $1$, we demonstrate the violation in \cite{tversky1992disjunction_effect_empirical} in Figure \ref{fig:violation_of_sure_thing_principle}. 
\begin{figure}
    \centering
    \includegraphics[scale=0.28]{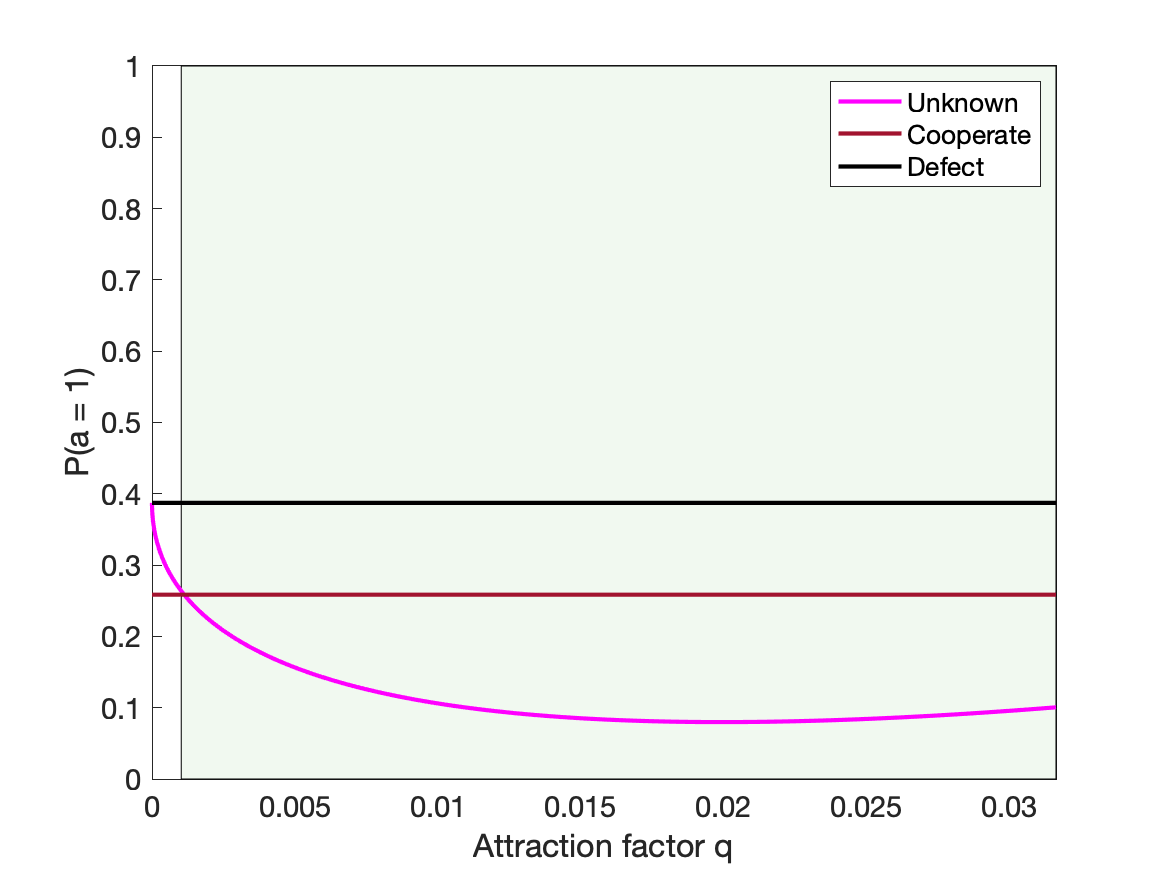}
    \caption{Demonstration of violation of sure-thing-principle with $\eps = 1$:   The probability of defection for the human receiver when the opponent is known to defect and cooperate is $0.39$ and $0.26$, respectively. The probability of defection when the action of the opponent is unknown depends on the attraction factor as illustrated in the dashed curve. We observe a violation of the total law of probability in our framework. Specifically, when the attraction factor exceeds $0.001$ (indicated by the purple region), the probability of defection for the human receiver is no longer a convex combination of the probabilities associated with the opponent's certain defection or cooperation.}
\label{fig:violation_of_sure_thing_principle}
\end{figure}


In addition, we plot in Figure \ref{fig:threshold_prior} the human detector's optimal decision rule, characterized by the threshold,  given the prior that the opponent chooses to defect $p(H_1)$. We observe that as $p(H_1)$ increases from $0$ to $1$, the optimal detecting threshold initially increases slowly, but later drastically after a certain point, indicating that the probability that human cooperates decreases slowly when the prior $p(H_1)$ is not too large, but switch to defect very quickly after a certain point. 
\begin{figure}
    \centering
    \includegraphics[scale=0.23]{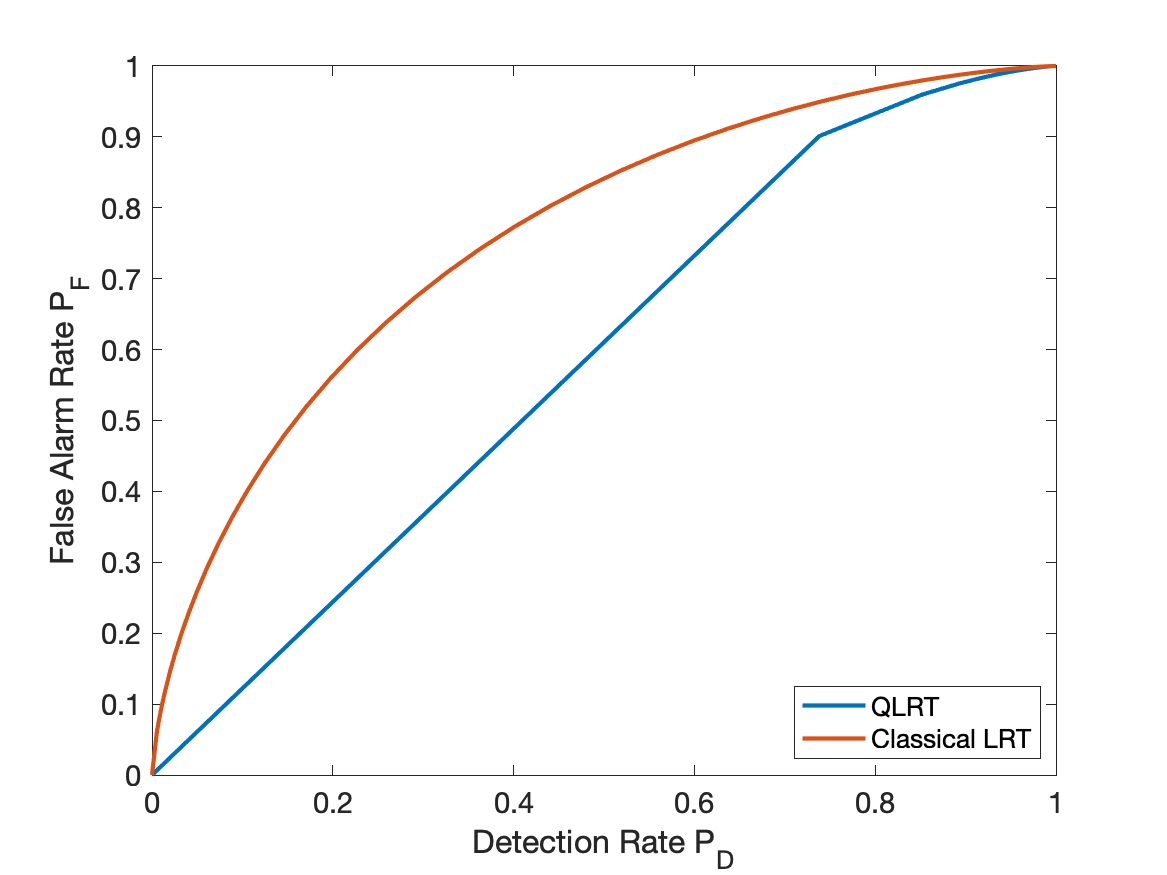}
    \caption{The ROC curves of the human agent's optimal decision rules using our framework and a previous prospect-theory-based hypothesis testing framework \cite{nadendla2016human_agent_hypo_testing}.    }
    \label{fig:threshold_prior}
\end{figure}

\textbf{The human agent's detection performance:}
We plot the receiver operating characteristic (ROC) curve of the human detector's optimal decision rule, which is determined by the threshold.
 Assume that the two underlying distributions $f_1,f_0$ are both Gaussian with mean $0$ and mean $1$ and with the same variance of $1$. We choose randomly the coefficients $a_{sk}$ in \eqref{eq:rho0_prospect_state} and \eqref{eq:rho1_prospect_state}.
Figure \ref{fig:threshold_prior} illustrates the observed ROC curve, highlighting the impact of the attraction factor on the human agent's detection performance. Notably, the quantum human agent exhibits a distinct behavior due to the influence of the attraction factor, which arises from the quantum interference of the human mind. This interference phenomenon leads to an intriguing deviation in the detection performance.

\section{Conclusion}
\label{sec:conclusion}

In this paper, we have proposed a novel detection framework for human-sensor systems based on quantum decision theory, which effectively captures the bounded rationality of human agents, including their risk preferences, in the decision-making process. We have specifically focused on deriving the optimal decision rule for human agents under a particular case. Additionally, we conduct an analysis to highlight the impact of the attraction factor on the detecting performance of human agents. To achieve this, we have compared the receiver operating characteristic (ROC) curves obtained using quantum decision models with those obtained using prospect-theory-based models.

Our framework possesses the flexibility to be extended to cases where the sender exhibits different preferences, which can be characterized by specifying distinct utility functions. For instance, the sender may be adversarial towards the human receiver and aim to employ persuasive strategies that lead the human agent to higher error rates. Such conflicting relationships often arise in network security detection problems and are frequently studied using game theory (see \cite{manshaei2013game_network_security}).

\newpage
\bibliographystyle{IEEEbib}
\bibliography{thesis}
\newpage

\end{document}